\documentclass{amsart}

\usepackage{amsthm}
\usepackage{amsmath}
\usepackage{amsfonts}
\usepackage{amssymb}
\usepackage{bbm}
\usepackage{hyperref}
\usepackage{cleveref}
\usepackage{tikz-cd}
\usepackage{enumitem}
\usepackage{mathtools}

\usepackage[norefs,nocites]{refcheck}

\makeatletter
\newcommand{\refcheckize}[1]{%
	\expandafter\let\csname @@\string#1\endcsname#1%
	\expandafter\DeclareRobustCommand\csname relax\string#1\endcsname[1]{%
		\csname @@\string#1\endcsname{##1}\wrtusdrf{##1}}%
	\expandafter\let\expandafter#1\csname relax\string#1\endcsname
}
\makeatother

\refcheckize{\cref}
\refcheckize{\Cref}

\newtheorem{proposition}{Proposition}
\newtheorem{theorem}{Theorem}
\newtheorem{corollary}{Corollary}
\newtheorem{lemma}{Lemma}
\theoremstyle{definition}

\newtheorem{example}{Example}
\theoremstyle{remark}
\newtheorem{remark}{Remark}

\newcommand{\range}[1]{\{1,\dots,#1\}}
\newcommand{\id}{\mathbbm{1}}
\newcommand{\card}[1]{|{#1}|}
\newcommand{\Q}{\mathbb{Q}}

\DeclareMathOperator{\Iso}{Iso}
\DeclareMathOperator{\Ker}{Ker}
\DeclareMathOperator{\F}{F}
\DeclareMathOperator{\Mon}{Mon}
\DeclareMathOperator{\restr}{restr}
\DeclareMathOperator{\Stab}{Stab}

\newcommand{\M}{M}
\renewcommand{\P}{\mathcal{P}}
\renewcommand{\O}{\mathcal{O}}

\newcommand{\defeq}{\vcentcolon=}

\title{Isometry Groups of Combinatorial Codes}
\author[S. Dyshko]{Serhii Dyshko}
\email{dyshko@univ-tln.fr}

\begin{document}
\begin{center}
	Draft, \today
\end{center}
\begin{abstract}
	Two isometry groups of combinatorial codes are described: the group of automorphisms and the group of monomial automorphisms, which is the group of those automorphisms that extend to monomial maps. Unlike the case of classical linear codes, where these groups are the same, it is shown that for combinatorial codes the groups can be arbitrary different. Particularly, there exist codes with the full automorphism group and the trivial monomial automorphism group. In the paper the two groups are characterized and codes with predefined isometry groups are constructed.
\end{abstract}

\maketitle

\section{Introduction}
An automorphism group of a classical linear code is the group of those linear bijections from the code to itself that preserve the Hamming distance.
MacWilliams proved in her Ph.D. thesis~\cite{macwilliams-phd61} that each linear Hamming isometry of a classical linear code extends to a monomial map, i.e., it acts by permutation of coordinates and multiplication of coordinates by nonzero scalars.
Consequently, each automorphism of a linear code extends to a monomial map.

As it was shown in numerous papers~\cite{dinh-lopez, greferath, wood-foundations}, for linear codes over module alphabets an analogue of the MacWilliams Extension Theorem does not hold in general. This means that there may exist codes with automorphisms that do not extend to monomial maps. In the context of combinatorial codes, i.e., codes without any algebraic structure, the situation is similar, see~\cite{aug1, kov1, heise}.

For the case of non-classical codes, along with the automorphism group of a code there is observed the subgroup of those automorphisms that extend to monomial maps. As was mentioned above, unlike the case of classical linear codes, the two groups may not be the same.

In~\cite{wood-iso-slides} Wood investigated the question of how different the two groups of a linear code over a matrix module alphabet can be. 
He showed, under certain assumptions, that there exists a linear code over a matrix module alphabet with predefined group of automorphism and group of monomial automorphism.
In this paper we prove a similar statement for combinatorial codes.

\section{Preliminaries}
Let $A$ be a finite set, called an \emph{alphabet}, and let $n$ be a positive integer.
The \emph{Hamming distance} is a function
$\rho_H: A^n \times A^n \rightarrow \{0,\dots,n\}$, defined as, for $a,b \in A^n$,
$$
\rho_H(a,b) \defeq \card{ \{ k \mid a_k \neq b_k \} }.
$$
The set $A^n$ equipped with the Hamming distance is a metric space called a \emph{Hamming space}.
A \emph{code} $C$ is a subset of the Hamming space $A^n$. The elements of $C$ are called \emph{codewords}.

A map $f:C \rightarrow A^n$ is called a \emph{Hamming isometry} if for each two codewords $a,b \in C$, $\rho_H(a,b) = \rho_H(f(a),f(b))$, i.e., the map $f$ preserves the Hamming distance.

For a finite set $X$, let $\mathfrak{S}(X)$ denote the symmetric group on $X$. Denote $\mathfrak{S}_n = \mathfrak{S}(\range{n})$, where $n$ is a positive integer.

A map $h: A^n \rightarrow A^n$ is called \emph{monomial} if there exists a permutation $\pi \in \mathfrak{S}_n$ and permutations $\sigma_1, \dots, \sigma_n \in \mathfrak{S}(A)$ such that for each $a \in A^n$,
$$
h(a) = \left( \sigma_1(a_{\pi(1)}), \dots, \sigma_n(a_{\pi(n)})\right).
$$

Let $C \subseteq A^n$ be a code with $m \geq 3$ codewords. Consider the set $\M = \range{m}$, called the \emph{set of messages}, and consider an \emph{encoding map} $\lambda : M \rightarrow A^n$ of the code $C$, i.e., an injective map such that $\lambda(\M) = C$. For every $g \in \mathfrak{S}(M)$ the map $\lambda g \lambda^{-1}: C \rightarrow C$ is a well-defined bijection.

Define \emph{the group of automorphisms} of $C$,
$$
\Iso(C) \defeq \{ g \in \mathfrak{S}(M) \mid \lambda g\lambda^{-1} \text{ is a Hamming isometry}\},
$$
and \emph{the group of monomial automorphisms} of $C$,
$$
\phantom{.}\Mon(C) \defeq \{ g \in \Iso(C) \mid \lambda g\lambda^{-1} \text{ extends to a monomial map} \}.
$$

\begin{remark}
	In coding theory the group $\Mon^*(C)$ of those monomial maps that preserve $C$ is often used.
	The set of all monomial maps of $A^n$ form a group, which is isomorphic to the wreath product $\mathfrak{S}_n \wr \mathfrak{S}(A)$, see~\cite[Section~2.6]{dixon}.
	Note that $\Mon(C)$ and $\Mon^*(C)$ are different objects: $\Mon(C)$ is a subgroup of $\mathfrak{S}(M)$ and $\Mon^*(C)$ is a subgroup of the full group of monomial maps.
	However, there exists a connection. For a monomial map $h \in \Mon^*(C)$ the map $\lambda^{-1} h\lambda$ is in $\Iso(C)$. By defining the map $\restr: \Mon^*(C) \rightarrow \Iso(C), h \mapsto \lambda^{-1} h\lambda,$
we have the equality of groups $\Mon(C) = \restr(\Mon^*(C))$. The groups $\Mon(C)$ and $\Mon^*(C)$ are isomorphic unless the homomorphism $\restr$ has a nontrivial kernel. The last holds if and only if the code has two columns that differ by a permutation of the alphabet.
\end{remark}

In~\cite{bonneau} and in~\cite[Theorem 1]{isomconf} the equality of the groups $\Iso(A^n)$ and $\Mon(A^n)$ was proven. However, in general, for codes in $A^n$ the equality of groups does not hold.

\begin{example}
	Suppose $A = \{0,1\}$, $\M = \range{5}$. Consider the code $C$ of cardinality $5$ in $A^4$,
$$
	C = \{ (0,0,0,0), (1,1,0,0), (1,0,1,0), (1,0,0,1), (0,1,1,0)\},
$$
	where the encoding $\lambda: M \rightarrow C$ is given in the presented order.
	The group $\Iso(C) < \mathfrak{S}(M)$ is generated by cycles $(1,2)$, $(1,2,3)$ and $(4,5)$ and has $12$ elements. From the other side, $\Mon(C)$ is a subgroup of $\Iso(C)$ generated by permutations $(1,2)$ and $(1,2,3)$. For example, denoting $g_0 = (4,5) \in \Iso(C)$, the map $\lambda g_0\lambda^{-1}$ is an automorphism of $C$ that does not extend to a monomial map,
$$
\begin{array}{cccc}
0 & 0 & 0 & 0 \\
1 & 1 & 0 & 0 \\
1 & 0 & 1 & 0 \\
1 & 0 & 0 & 1 \\
0 & 1 & 1 & 0
\end{array} \xrightarrow{\lambda g_0 \lambda^{-1}}
\begin{array}{cccc}
0 & 0 & 0 & 0 \\
1 & 1 & 0 & 0 \\
1 & 0 & 1 & 0 \\
0 & 1 & 1 & 0 \\
1 & 0 & 0 & 1
\end{array}.
$$
Indeed, one can easily check that $\lambda g_0 \lambda^{-1} : C \rightarrow C$ is a Hamming isometry. Also, the fourth column from the right hand side has equal first four elements, but there is no such column from the left hand side. Hence, the code automorphism $\lambda g_0 \lambda^{-1}$ does not extend to a monomial map.

Therefore, $\Mon(C) \neq \Iso(C)$. Note that the group $\Mon^*(C)$ is generated by two elements. One acts by swapping the second and the third column of $C$ and another acts by inverting the symbols in the first two columns and then swapping the first two columns. Both groups $\Mon^*(C)$ and $\Mon(C)$ have $6$ elements and are isomorphic to the symmetric group $\mathfrak{S}_3$. 
\end{example}
There exist more complex examples. Recall that a code $C \subseteq A^n$ is an $(n, K, d)$ $q$-ary code if it has cardinality $K$, minimum distance $d$ and the alphabet has $q$ elements. 
According to \cite{heise}, if $C$ is $(q, q^2, q - 1)$ or $(q + 1, q^2, q)$ $q$-ary MDS code with $q \neq 2$, then $\card{\Iso(C)} > \card{\Mon(C)}$, and thus $\Iso(C) \neq \Mon(C)$. The same holds, for example, for $(q, (q-1)^2, q-1)$ $q$-ary equidistant codes, where $q \geq 5$ and both $q$ and $q-1$ are prime powers (see \cite{kov1}).

The main result of this paper is given in the next section.

\section{Main result}

Let $G$ be a group acting on a set $X$. We say that a subgroup $H \leq G$ is \emph{closed under the action on $X$} if $H$ consists of all those elements in $G$ that preserve the orbits of $H$.
To be precise, we say that $H$ is closed under the action on $X$ if,
$$
	\forall g \in G,\;\forall O \in X/H,\quad g(O) = O \Rightarrow g \in H,
$$
where $X/H$ is the set of orbits of $H$ acting on $X$.
A similar definition of a closed group and a definition of a closure of a group were introduced by Wielandt in his 2-closure theory (see~\cite{wielandt_1969}). More results on the closure can be found in~\cite[Section~2.4]{dixon} and \cite{xu_closures}.

Let $q$ denote the cardinality of the alphabet $A$, $q \geq 2$.
Consider the set $\P$ of all the partitions of the set $\M$ that have at most $q$ classes,
$$
\P \defeq \big\lbrace \{ c_1, \dots, c_t \} \mid c_1 \sqcup \dots \sqcup c_t = \M,\;\; t \leq q \big\rbrace,
$$
where $c_i \subseteq \M$, for $i \in \range{t}$, and $\sqcup$ denotes the disjoint union of sets.

The canonical action of the group $\mathfrak{S}(M)$ on the set $\P$ is defined in the following way, for $g \in \mathfrak{S}(M)$, for $\alpha = \{c_1, \dots, c_t \} \in \P$,
$$
g(\alpha) \defeq \{ g(c_1), \dots, g(c_t) \}.
$$

In $\P$ we distinguish a subset
$$
\P_2 \defeq \Bigl\lbrace \big\lbrace \{ i,j\}, \{\M\setminus \{i,j \} \} \big\rbrace \Bigm| \{i,j\} \subset \M \Bigr\rbrace.
$$
Since each partition in $\P_2$ has two classes, which is not greater than $q$, $\P_2 \subset \P$.
The group $\mathfrak{S}(M)$ naturally acts on $\P_2$ and $\P \setminus \P_2$.

\begin{theorem}\label{thm:main}
	Let $A$ be a finite set alphabet of cardinality $q \geq 2$ and
	let $C$ be a code over the alphabet $A$ of cardinality $m \geq 5$ or $m = 3$. The following statements hold.
\begin{enumerate}[label=(\roman*),ref=(\roman*)]
		\item The group $\Iso(C) \leq \mathfrak{S}(M)$ is closed under the action on $\P_2$.\label{thm:main:a}
		\item The group $\Mon(C)$ is equal to an intersection of $\Iso(C)$ with a subgroup of $\mathfrak{S}(M)$ closed under the action on $\P\setminus \P_2$. \label{thm:main:b}
		\item For each subgroup $H_1 \leq \mathfrak{S}(M)$ that is closed under the action on $\P \setminus \P_2$, for each subgroup $H_2 \leq \mathfrak{S}(M)$ that is closed under the action on $\P_2$, there exists a code $C$ of cardinality $m \geq 5$ such that
		$$
		\Mon(C) = H_1 \cap \Iso(C) \quad \text{and} \quad \Iso(C) = H_2.
		$$\label{thm:main:c}
\end{enumerate}
\end{theorem}
From the fact that, for $m \geq 5$, the trivial subgroup $\{e\}$ and the full group $\mathfrak{S}(M)$ are closed under the action on $\P \setminus \P_2$ and $\P_2$ respectively, we get the first corollary.
\begin{corollary}
For every integer $m \geq 5$, there exists a code $C$ of cardinality $m$ with the maximal group $\Iso(C) = \mathfrak{S}(M)$ and the minimal group $\Mon(C) = \{e\}$.
\end{corollary}
By using \Cref{thm:main} \ref{thm:main:c} with $H_1 = \mathfrak{S}(M)$, which is closed under the action on $\P \setminus \P_2$, we get the second corollary.
\begin{corollary}
	For every integer $m \geq 5$ or $m = 3$, for each subgroup $H \leq \mathfrak{S}(M)$ that is closed under the action on $\P_2$, there exists a code $C$ of cardinality $m$ such that the groups $\Mon(C)$ and $\Iso(C)$ coincide and are both equal to $H$.
\end{corollary}

\begin{example}\label{ex:comb:ex1}
	For $m = 5$ and $q = 2$, consider the code $C$ of the following form.
	$$	
	\begin{array}{c|ccccc|cccccccccc}
	0 & 0& 1& 2& 3& 4& 6& 5& 4& 3& 4& 3& 2& 2& 1& 0 \\
	\hline
	0 & 1& 0& 0& 0& 0& 1& 1& 1& 1& 0& 0& 0& 0& 0& 0 \\
	0 & 0& 1& 0& 0& 0& 1& 0& 0& 0& 1& 1& 1& 0& 0& 0 \\
	0 & 0& 0& 1& 0& 0& 0& 1& 0& 0& 1& 0& 0& 1& 1& 0 \\
	0 & 0& 0& 0& 1& 0& 0& 0& 1& 0& 0& 1& 0& 1& 0& 1 \\
	0 & 0& 0& 0& 0& 1& 0& 0& 0& 1& 0& 0& 1& 0& 1& 1 \\
	\end{array}
	$$
The numbers over the header line represent the number of occurrences of the column under the line in the code. For instance, in this example, the third column appears once in the code, the fourth column appears twice in the code and the second column does not appear anywhere in the code.

We indicated several columns that does not appear in the code for the completeness. In fact, all possible columns (up to a permutation of the alphabet) are presented in the table. The vertical lines separate the columns with different numbers of $1$s.

The code $C$ is a $(40, 5, 22)$ equidistant binary code with the maximal automorphism group $\Iso(C) = \mathfrak{S}(M)$ and the minimal monomial automorphism group $\Mon(C) = \{e\}$. 
\end{example}

\begin{remark}
In general, for $m \geq 5$, there is no direct relation between the subgroups of $\mathfrak{S}(\M)$ closed under the action on $\P_2$ and $\P\setminus \P_2$.
There exists a subgroup that is closed under the action on $\P \setminus \P_2$ but not closed under the action on $\P_2$. For example, 
if $m = 5$, $q = 3$ and the group is $G = \langle (1,2)(3,4), (1,2)(3,5)\rangle < \mathfrak{S}_5$.
	There also exists a subgroup that is closed under the action on $\P_2$ but not closed under the action on $\P \setminus \P_2$. Consider $m = 5$, $q = 2$ and $G = \langle (1,2)(3,4) \rangle < \mathfrak{S}_5$.
\end{remark}

For codes with $m = 4$ codewords the statement of the theorem is not correct in general and needs to be refined. This case is observed in \Cref{sec:comb:m-equal-4}.

\section{Multiplicity function}
For sets $X, Y$, let $\F(X,Y)$ denote the set of all maps from $X$ to $Y$.
Consider the map $\Psi: \F(\M,A) \rightarrow \P$,
$$
x \mapsto \Psi(x) \defeq \{ x^{-1}(a) \mid a \in A\} \setminus \{ \emptyset \}.
$$
The number of classes in $\Psi(x)$ is at most $\card{A} = q$, so $\Psi(x) \in \P$ and hence the map is defined.

Recall that $\lambda: \M \rightarrow A^n$ is an encoding map of $C$. Let $\lambda_k: \M \rightarrow A$ denote the projection of $\lambda$ on $k$th coordinate for $k \in \range{n}$.
Define the \emph{multiplicity function} $\eta_\lambda \in \F(\P, \Q)$, as follows, for $\alpha \in \P$,
$$
\eta_\lambda(\alpha) \defeq \card{ \{k \mid \Psi({\lambda_k}) = \alpha \} }.
$$

\begin{proposition}\label{thm:from-eta-to-parametrization}
	For every non-zero function $\eta \in \F(\P, \Q)$ with nonnegative integer values there exist a positive integer $n$ and a map $\lambda: \M \rightarrow A^n$, such that
	$\eta_\lambda = \eta.$
\end{proposition}
\begin{proof}
	Define $n = \sum_{\alpha \in \P} \eta(\alpha)$ and let $\alpha_1, \dots, \alpha_n \in \P$ be the $n$-tuple of partitions such that for all $\alpha \in \P$,
	$$ \eta(\alpha) = \card{ \{ k \mid \alpha = \alpha_k \} }.$$

	Enumerate the elements of the alphabet $A = \{a_1, \dots, a_q\}$ and fix $k \in \range{n}$. Let $\alpha_k = \{ c_1, \dots, c_t \}$, for some $t \leq q$. Define the map $\lambda_k \in \F(\M,A)$ as
	$$
	\forall i \in \range{t}, \forall j \in c_i, \;\; \lambda_k(j) = a_i.
	$$
	It is easily seen that $\Psi(\lambda_k) = \alpha_k$. Define $\lambda = (\lambda_1,\dots, \lambda_n) : \M \rightarrow A^n$.
	Then, for all $\alpha \in \P$,
$\eta_\lambda(\alpha) = \card{ \{k \mid \Psi(\lambda_k) = \alpha \} } = \card{ \{k \mid \alpha_k = \alpha \} } =\eta(\alpha)$.
\end{proof}

\section{Extension criterion and stabilizers}
Let $\O$ denote the set of pairs of elements of $M$,
$$
\O \defeq \big\lbrace \{i,j\}  \bigm|  \{i,j\} \subset \M \big\rbrace.
$$ The group $\mathfrak{S}(M)$ acts on $\O$ in the following way, for $g \in \mathfrak{S}(M)$,
$$
g(\{i,j\} ) \defeq \{ g(i), g(j) \}. 
$$
Consider the action of $\mathfrak{S}(M)$ on the vector spaces $\F(\P,\Q)$ and $\F(\O,\Q)$, so that, for $g \in \mathfrak{S}(M)$, for $\eta \in \F(\P,\Q)$ and $\alpha \in \P$,
$$
g(\eta)(\alpha) \defeq \eta(g^{-1}(\alpha)),
$$
and for $x \in \F(\O,\Q)$ and $p \in \O$,
$$
g(x)(p) \defeq x(g^{-1}(p)).
$$
The group $\mathfrak{S}(M)$ acts on the vector spaces by automorphisms, i.e., a map $\phi: V \rightarrow V$, associated to the action of $g$ on $V$, is a $\mathbb{Q}$-linear bijection, where $V$ is either $\F(\P, \Q)$ or $\F(\O, \Q)$.

For $\alpha \in \P$ and $p = \{i,j\} \in \O$, define the function $\Delta_\alpha \in \F(\O,\Q)$ as
$$
\Delta_\alpha (p) \defeq \left\lbrace \begin{array}{ll}
0,& \text{if $i$ and $j$ belong to the same class in $\alpha$};\\
1,& \text{otherwise}.
\end{array}
\right.
$$

Consider the $\mathbb{Q}$-linear map,
$$
W : \F(\P, \Q) \rightarrow \F(\O, \Q), \quad W(\eta)(p) \defeq \sum_{\alpha \in \P} \eta(\alpha) \Delta_\alpha{(p)},
$$
where $\eta \in \F(\P,\Q)$, $p \in \O$. A similar map was observed in~\cite{wood-structure-of-linear-codes, wood-foundations}.

\begin{proposition}\label{thm:comb-extension-criterion}
	A map $f: C \rightarrow A^n$ is a Hamming isometry if and only if $W(\eta_\lambda) = W(\eta_{f \lambda})$. The map $f$ extends to a monomial map if and only if $\eta_\lambda = \eta_{f \lambda}$.
\end{proposition}
\begin{proof}
Calculate the Hamming distance, for all $ p = \{i,j\} \in \O$,
\begin{align*}
	\rho_H(\lambda(i), \lambda(j))
	&=
	\card{ \{ k \mid \lambda_k(i) \neq \lambda_k(j) \} }\\
	&=
	\card{ \{ k \mid \Delta_{\Psi(\lambda_k)}(p)=1 \} }
	\\
	&=\sum_{\alpha \in \P} \card{ \{ k \mid \Psi(\lambda_k) = \alpha \} } \Delta_{\alpha}(p)\\
	&=
	\sum_{\alpha \in \P} \eta_\lambda(\alpha) \Delta_{\alpha}(p) = W(\eta_\lambda)(p).
\end{align*}
The map $f$ is a Hamming isometry if and only if for all $\{i,j\} \subset \M$,
	$$
	\rho_H(\lambda(i), \lambda(j)) = \rho_H(f \lambda(i), f \lambda(j)),
	$$
which is equivalent to the equality $W(\eta_\lambda) = W(\eta_{f \lambda})$.
	
The map $f$ extends to a monomial map if and only if there exist a permutation $\pi \in \mathfrak{S}_n$ and permutations $\sigma_1, \dots, \sigma_n \in \mathfrak{S}(A)$ such that $f\lambda_{k} = \sigma_k\lambda_{\pi(k)}$, for all $k \in \range{n}$.

It is an easy exercise to verify that $\Psi(x) = \Psi(y)$ for two maps $x, y \in \F(\M, A)$, if and only if there exists a permutation $\sigma \in \mathfrak{S}(A)$ such that $\sigma x = y$.

If $f$ extends to a monomial map, then for all $\alpha \in \P$,
\begin{align*}
\eta_{f\lambda}(\alpha) &= \card{ \{k \mid \Psi({f\lambda_k}) = \alpha \} }\\
&= \card{ \{k \mid \Psi({\sigma_k \lambda_{\pi(k)}}) = \alpha \} }\\
&= \card{ \{k \mid \Psi(\lambda_{\pi(k)}) = \alpha \} } = \eta_\lambda(\alpha).
\end{align*}
Conversely, let $\eta_{f\lambda} = \eta_\lambda$. Then for all $\alpha \in \P$ the cardinality of sets $X_{\alpha} = \{k \mid \Psi(\lambda_k) = \alpha \}$ and $Y_{\alpha} = \{k \mid \Psi(f \lambda_k) = \alpha \}$ are equal. The set $\{1,\dots, n\}$ is then a disjoint union of the subsets $X_\alpha$ for $\alpha \in \P$. It is also equal to a disjoint union of the subsets $Y_\alpha$ for $\alpha \in \P$.
Thus, there exists $\pi \in \mathfrak{S}_n$ such that for all $\alpha \in \P$, $\pi(X_\alpha) = Y_\alpha$. Therefore, for all $k \in \range{n}$, $\Psi(\lambda_{\pi(k)}) = \Psi(f\lambda_k)$. 
From this, for every $k \in \range{n}$ there exists $\sigma_k \in \mathfrak{S}(A)$ such that $\sigma_k\lambda_{\pi(k)} = f \lambda_k$, which means that $f$ extends to a monomial map.
\end{proof}

For a map $\eta \in \F(\P, \Q)$ define two stabilizers,
\begin{align*}
\Stab(\eta) &\defeq \{ g \in \mathfrak{S}(M) \mid g(\eta) = \eta \},\\
\Stab_{W}(\eta) &\defeq \{ g \in \mathfrak{S}(M) \mid W(g(\eta)) = W(\eta) \}.
\end{align*}

\begin{proposition}\label{thm:groups-as-stabs}
If $C$ is a code with an encoding map $\lambda : M \rightarrow A^n$, then
$$
\Mon(C) = \Stab(\eta_\lambda), \quad \Iso(C) = \Stab_W(\eta_\lambda).
$$
\end{proposition}
\begin{proof}
For all $g \in \mathfrak{S}(M)$, for all $\alpha \in \P$,	
\begin{align*}
	\eta_{\lambda g}(\alpha) = 
	\card{ \{ k \mid \Psi(\lambda_k g) = \alpha \} } = 
	\card{ \{ k \mid \Psi(\lambda_k) = g^{-1}(\alpha) \} } =\eta_\lambda(g^{-1}(\alpha)) = g(\eta_\lambda)(\alpha).
\end{align*}
The statement of the proposition follows directly from \Cref{thm:comb-extension-criterion} by substituting $\lambda g \lambda^{-1}$ in place of $f$.
\end{proof}

\section{Two matrices}
\newcommand{\mysymp}{R}
\newcommand{\rel}[2]{\card{#1 \cap #2} = 1}
\newcommand{\nrel}[2]{\card{#1 \cap #2} \neq 1}
\newcommand{\zrel}[2]{\card{#1 \cap #2} = 0}
\newcommand{\trel}[2]{\card{#1 \cap #2} = 2}

For two pairs $p, t \in \O$, the intersection $p \cap t \subset \M$ can have at most $2$ elements. On $\O$ fix an order.
Let $B$ be the $\card{\O} \times \card{\O}$ matrix defined over $\Q$ and indexed by the elements of $\O$. For $p, t \in \O$,
$$
B_{p,t} \defeq \left\lbrace \begin{array}{ll}
1,& \text{if } \rel{p}{t};\\
0,& \text{if } \nrel{p}{t}.
\end{array}\right.
$$

Let $m \geq 5$ or $m = 3$. Define the $\card{\O} \times \card{\O}$ matrix $D$ as follows, for $p,t \in \O$,
$$
2(m-2)(m-4) \times D_{p,t} \defeq \left\lbrace \begin{array}{ll}
-m^2 + 8m - 14,& \text{if } \card{p \cap t} = 2 \; (\iff p = t);\\
m-4,& \text{if } \rel{p}{t};\\
-2,& \text{if } {\zrel{p}{t}}\; ( \iff p \cap t = \emptyset).
\end{array}\right.
$$

\begin{lemma}\label{lemma:Binvertible}
	If $m \geq 5$ or $m = 3$, then $BD = DB = I$, where $I$ is the identity $\card{\O} \times \card{\O}$ matrix.
\end{lemma}
\begin{proof}
	
	For $p, t\in \O$,
	$$
	(BD)_{p,t} = \sum_{r \in \O} B_{p,r} D_{r,t}
	= B_{p,t} D_{t,t} + \sum_{\substack{\rel{r}{p}\\ \rel{r}{t}}} D_{r,t} + \sum_{\substack{\rel{r}{p}\\ \zrel{r}{t}}} D_{r,t}.
	$$
	If $p = t$, then
	$$
	(BD)_{p,p} = 0+\sum_{\rel{r}{p}} D_{r,p}+0 = \frac{1}{2(m-2)}
	\card{ \{ r \in \O \mid \rel{r}{p} \} }  = 1.
	$$
	If $\rel{p}{t}$, then
	\begin{align*}
	(BD)_{p,t} 
	= \frac{-m^2 + 8m - 14}{2(m-2)(m-4)} + \frac{1}{2(m-2)} \card{ \{r \in \O \mid \rel{r}{p};\; \rel{r}{t} \} }
	\\+ \frac{-1}{(m-2)(m-4)}
	\card{ \{ r \in \O \mid \rel{r}{p};\; \zrel{p}{t}\} }\\
	= \frac{-m^2 + 8m - 14}{2(m-2)(m-4)} + \frac{(m-4)}{2(m-2)(m-4)} (m-2) + \frac{-2}{2(m-2)(m-4)} (m-3) = 0.
	\end{align*}
	If $\zrel{p}{t}$, then
	\begin{align*}
	(BD)_{p,t}
	= 0 + \frac{1}{2(m-2)} \card{ \{r \in \O \mid \rel{r}{p};\; \rel{r}{t} \} }
	\\+ \frac{-1}{(m-2)(m-4)}
	\card{ \{ r \in \O \mid \rel{r}{p};\; \zrel{r}{t}\} }\\
	= \frac{4}{2(m-2)} + \frac{-2}{2(m-2)(m-4)} 2(m-4) = 0.
	\end{align*}
Hence, $BD = I$, the matrices $B$ and $D$ are invertible and $B^{-1} = D$.
\end{proof}

\section{Properties of the map W}

For a partition $\alpha \in \P$ let $\id_\alpha\in\F(\P,\Q)$ be the map defined as follows, for $\beta \in \P$,
$$
\id_\alpha(\beta) \defeq \left\lbrace \begin{array}{ll}
1,& \text{if } \beta = \alpha;\\
0,& \text{if } \beta \neq \alpha.\\
\end{array} \right.
$$
Note that, for $\alpha \in \P$, for $p \in \O$,
\begin{equation}\label{eq:comb:w-id-delta}
W(\id_\alpha)(p) = \sum_{\beta \in \P} \id_\beta(\alpha)\Delta_\beta(p) = \Delta_\alpha(p).
\end{equation}

From now we assume that $m \geq 5$ or $m = 3$.
For every map $x \in \F(\O,\Q)$ define the map in $\F(\P,\Q)$,
$$
\xi_x \defeq \sum_{r \in \O} x(r) \sum_{p \in \O}  D_{r,p}  \id_{\{p, \M\setminus p\}},
$$
where $D$ is the matrix defined previously.
Using \Cref{lemma:Binvertible}, for $x \in \F(\O,\Q)$ and $t\in \O$,
\begin{align}\label{eq:comb:w-eta-x}
\begin{split}
W(\xi_x)(t) & 
\overset{(\ref{eq:comb:w-id-delta})}{=}
\sum_{r \in \O} x(r) \sum_{p \in \O} D_{r,p} \Delta_{\{p, \M\setminus p\}}(t)\\
&= \sum_{r \in \O} x(r) \sum_{p \in \O}  D_{r,p} B_{p,t} =\sum_{r \in \O} x(r) (DB)_{r,t} = x(t).
\end{split}
\end{align}

For every $\alpha \in \P \setminus \P_2$ define the function in $\F(\P, \Q)$,
$$
	\zeta_\alpha \defeq \id_\alpha - \xi_{\Delta_\alpha}.
$$
Recall that for a partition $\alpha \in \P$, the function $\Delta_\alpha$ is in $\F(\O,\Q)$.
Using this equality, for $\alpha \in \P$, 
\begin{equation}\label{eq:comb:w-eta-alpha}
W(\zeta_\alpha)
\overset{(\ref{eq:comb:w-id-delta})}{=}
\Delta_\alpha - W(\xi_{\Delta_\alpha})
\overset{(\ref{eq:comb:w-eta-x})}{=}
\Delta_\alpha - \Delta_\alpha = 0.
\end{equation}


Consider the subspace $V_0$ of $\F(\P, \Q)$ of those functions that take zero values on the partitions in $\P \setminus \P_2$,
$$
V_0 \defeq \{ \eta \in \F(\P, \Q) \mid \forall \alpha \in \P \setminus \P_2,\; \eta(\alpha) = 0\}.
$$
The subspace $V_0$ is the image of the canonical embedding of $\F(\P_2, \Q)$ into $\F(\P, \Q)$. 

\begin{proposition}\label{thm:sum-of-spaces}
	If $m \geq 5$ or $m = 3$, then $\F(\P, \Q) = V_0 \oplus \Ker W$.
\end{proposition}
\begin{proof}
	From \cref{eq:comb:w-eta-x}, the map $\xi_x \in \F(\P,\Q)$ is a pre-image of a map $x \in \F(\O,\Q)$ under the map $W$, and hence $W: \F(\P,\Q) \rightarrow \F(\O,\Q)$ is onto. The dimension of the kernel of $W$ is equal to $\card{\P} - \card{\O} = \card{\P} - \card{\P_2} = \card{\P \setminus \P_2}$, which is true for $m \neq 4$. 
	Obviously, the maps $\zeta_\alpha$ for $\alpha \in \P\setminus\P_2$, are linearly independent over $\mathbb{Q}$ and thus form a basis of $\Ker W$: see \cref{eq:comb:w-eta-alpha}.
	The dimension of $V_0$ is equal to $\card{\P_2}$. 
	Also, $\Ker W \cap V_0 = \{0\}$. Therefore, $\F(\P, \Q) = V_0 \oplus \Ker W$.
\end{proof}

\section{Proof of the main theorem}
Before starting to prove the main theorem, let us prove several necessary equalities.

For $g \in \mathfrak{S}(M)$ and $p=\{i,j\} \in \O$, the value $g(\Delta_\alpha)(p) = \Delta_\alpha(g^{-1}(p))$ is $0$ if $g^{-1}(i)$ and $g^{-1}(j)$ belong to the same class in $\alpha$, and the value is $1$ otherwise. Obviously, $g^{-1}(i)$ and $g^{-1}(j)$ are in the same class of $\alpha$ if and only if $i$ and $j$ are in the same class of $g(\alpha)$. Thus we have the equality 
\begin{equation}\label{eq:comb:delta-g}
\Delta_{g(\alpha)} = g(\Delta_{\alpha}).
\end{equation}

For $g \in \mathfrak{S}(M)$, for $\eta \in \F(\P, \Q)$ and $p \in O$,
\begin{align}\label{eq:comb:w1}
\begin{split}
W(g(\eta))(p) &= \sum_{\alpha \in \P} g(\eta)(\alpha) \Delta_{\alpha}(p)
=\sum_{\alpha \in \P} \eta(g^{-1}(\alpha)) \Delta_{\alpha}(p) \\
&=\sum_{\beta \in \P} \eta(\beta) \Delta_{g(\beta)}(p) \overset{(\ref{eq:comb:delta-g})}{=}
\sum_{\beta \in \P} \eta(\beta) \Delta_{\beta}(g^{-1}(p)) =
W(\eta)(g^{-1}(p)).
\end{split}
\end{align}

For $g \in \mathfrak{S}(M)$ and for all $\beta \in \P$, the value $g(\id_\alpha)(\beta) = \id_\alpha(g^{-1} (\beta))$ is $1$, if $g(\alpha) = \beta$, and is $0$ otherwise. Hence we have the equality
\begin{equation}\label{eq:comb:indicator-g}
g(\id_\alpha) = \id_{g(\alpha)}.
\end{equation}

For $g \in \mathfrak{S}(M)$, for $x \in \F(\O,\Q)$,
\begin{align}\label{eq:comb:g-xi-x}
\begin{split}
g(\xi_x) 
&= \sum_{r \in \O}  x(r) \sum_{p \in \O} D_{r, p} \id_{\{g(p),\M\setminus g(p)\}}
=\sum_{r \in \O} x(r) \sum_{g^{-1}(p) \in \O}  D_{r, g^{-1}(p)} \id_{\{p,\M\setminus p\}}\\
&= \sum_{r \in \O} x(r) \sum_{p \in \O}  D_{g(r), p} \id_{\{p,\M\setminus p\}}
= \sum_{g^{-1}(r) \in \O} x(g^{-1}(r)) \sum_{p \in \O}  D_{r, p} \id_{\{p,\M\setminus p\}} \\
&= \sum_{r \in \O}g(x)(r) \sum_{p \in \O}  D_{r, p} \id_{\{p,\M\setminus p\}} = \xi_{g(x)}.
\end{split}
\end{align}

For $g \in \mathfrak{S}(M)$ and $\alpha \in \P$, 
\begin{equation}\label{eq:comb:eta-alpha-g}
g(\zeta_\alpha) = g(\id_{g(\alpha)}) - g(\xi_{\Delta_\alpha})
= 
\id_{g(\alpha)} - \xi_{\Delta_{g(\alpha)}} = \zeta_{g(\alpha)},
\end{equation}
where the equality in the middle holds due to \cref{eq:comb:delta-g}, \cref{eq:comb:indicator-g} and \cref{eq:comb:g-xi-x}.

\begin{lemma}\label{lemma:sum-of-maps-stabs}
 If $\eta_0 \in \Ker W$ and $\eta_1 \in V_0$, then the equalities hold,
\begin{align*}
\Stab(\eta_0 + \eta_1) &= \Stab(\eta_0) \cap \Stab(\eta_1),\\
\Stab_W(\eta_0 + \eta_1) &= \Stab_W(\eta_1) = \Stab(\eta_1).
\end{align*}
\end{lemma}
\begin{proof}
We first show that the spaces $\Ker W$ and $V_0$ are invariant under the action of $\mathfrak{S}(M)$. Indeed, for $g \in \mathfrak{S}(M)$, if $\eta \in \Ker W$, then for all $p \in \O$,
$$
W(g(\eta))(p) \overset{(\ref{eq:comb:w1})}{=} W(\eta)(g^{-1}(p)) = 0,$$ and hence $g(\eta) \in \Ker W$. If $\eta \in V_0$, then for all $\alpha \in \P\setminus \P_2$, $g(\eta)(\alpha) = \eta(g^{-1}(\alpha)) = 0$, and thus $g(\eta) \in V_0$.

Now we prove the first equality. If $g \in \Stab(\eta_0) \cap \Stab(\eta_1)$, then $g(\eta_0) = \eta_0$ and $g(\eta_1) = \eta_1$. Hence $g(\eta_0 + \eta_1) = g(\eta_0) + g(\eta_1) = \eta_0 + \eta_1$ and thus $g \in \Stab(\eta_0 + \eta_1)$. Conversely, if $g \in \Stab(\eta_0 + \eta_1)$, then $g(\eta_0 + \eta_1) = \eta_0 + \eta_1$. Since $g(\eta_0) \in \Ker W$ and $g(\eta_1) \in V_2$, by the uniqueness of the decomposition (see \Cref{thm:sum-of-spaces}) $g(\eta_0) = \eta_0$ and $g(\eta_1) = \eta_1$, which means $g \in \Stab(\eta_0) \cap \Stab(\eta_1)$.
	Therefore, $\Stab(\eta_0) \cap \Stab(\eta_1) = \Stab(\eta_0 + \eta_1)$.
	
The second equality follows from the fact that for all $p \in \O$, $W(\eta_0 +  \eta_1) = W(\eta_0) + W(\eta_1) = W(\eta_1)$ and for all $g \in \mathfrak{S}(M)$, 
$W(g(\eta_0 +  \eta_1)) = W(g(\eta_0)) + W(g(\eta_1)) = W(g(\eta_1))$. 
	
Finally, we prove the third equality. The restriction of $W$ on the subspace $V_0$ is a bijection by \Cref{thm:sum-of-spaces}.
If $g \in \Stab_W(\eta_1)$, then $W(g(\eta_1)) = W(\eta_1)$.
The map $g(\eta_1)$ is in $V_0$.
	Hence, applying $W^{-1}$ to both sides of equality, we get $g(\eta_1) = \eta_1$ and thus $g \in \Stab(\eta_1)$. Since $\Stab(\eta_1) \leq \Stab_W(\eta_1) \leq \Stab(\eta_1)$, we get the third equality in the statement.
\end{proof}
Recall that $\P_2$ is a set of partitions of the form $\lbrace p, \M \setminus p \rbrace$, $p \in \O$.
Define the map $\id_{\P_2} \in \F(\P, \Q)$,
$$
\id_{\P_2} = \sum_{\alpha \in \P_2} \id_{\alpha}.
$$
\begin{lemma}\label{lemma:plus-const}
	Let $\eta \in \F(\P, \Q)$. For all $c \in \Q$,
\begin{align*}
\Stab(\eta) &= \Stab(\eta + c\id_{\P_2}),\\
	\Stab_W(\eta) &= \Stab_W(\eta + c\id_{\P_2}).
\end{align*}
\end{lemma}
\begin{proof}
For all $g \in \mathfrak{S}(M)$,
$$
g(\id_{\P_2}) = \sum_{\alpha \in \P_2} g(\id_{\alpha}) 
\overset{(\ref{eq:comb:indicator-g})}{=}
\sum_{\alpha \in \P_2} \id_{g(\alpha)} = \sum_{g^{-1}(\alpha) \in \P_2} \id_{\alpha} = \id_{\P_2}.
$$
We now prove the second equality of the statement. If $g \in \Stab_W(\eta)$, then
	\begin{align*}
	W(g(\eta + c\id_{\P_2})) 
	= W(g(\eta)) + W(cg(\id_{\P_2})) = W(\eta) + W(c\id_{\P_2}) = W(\eta + c\id_{\P_2}).
	\end{align*}
	Hence $g \in \Stab_W(\eta + c\id_{\P_2})$ and therefore $\Stab_W(\eta) \subseteq \Stab_W(\eta + c\id_{\P_2})$. From this, $\Stab_W(\eta + c\id_{\P_2}) \subseteq \Stab_W((\eta + c\id_{\P_2}) + (-c)\id_{\P_2})) = \Stab_W(\eta)$.
	The first equality is proven in the same way.
\end{proof}
Now we are ready to prove the main theorem of the paper.
\begin{proof}[Proof of \Cref{thm:main}]
Part \ref{thm:main:a}.
Let $\lambda \in \F(\M, A^n)$ be an encoding map of the code $C$.
From \Cref{thm:groups-as-stabs}, we have to show that $\Iso(C) = \Stab_W(\eta_\lambda)$ is closed under the action on $\P_2$. Since the bijection $\O \rightarrow \P_2$, $p \mapsto \{p,\M\setminus p\}$ preserves the action of the group $\mathfrak{S}(M)$, we have to show that $\Stab_W(\eta_\lambda)$ is closed under the action on $\O$.
If $g \in \mathfrak{S}(M)$ preserves the orbits of $\Stab_W(\eta_\lambda)$ on $\O$, then so does $g^{-1}$, and for all $p \in \O$,
$$
W(g(\eta_\lambda))(p) \overset{(\ref{eq:comb:w1})}{=} W(\eta_\lambda)(g^{-1}(p)) = W(\eta_\lambda)(p).
$$
Hence $g \in \Stab_W(\eta_\lambda)$, which means that $\Stab_W(\eta_\lambda)$ is closed under the action on $\O$. 

\vspace{1em}
Part \ref{thm:main:b}. Let $\lambda \in \F(\M, A^n)$ be an encoding map of the code $C$.
	From \Cref{thm:sum-of-spaces}, there exists the unique decomposition of the multiplicity function $\eta_\lambda \in \F(\P, \Q)$ into the sum of two maps $\eta_0 \in \Ker W$ and $\eta_1 \in V_0$ with $\eta_\lambda = \eta_0 + \eta_1$.

	From \Cref{thm:groups-as-stabs}, $\Stab(\eta_\lambda) = \Mon(C)$ and $\Stab_W(\eta_\lambda) = \Iso(C)$. From \Cref{lemma:sum-of-maps-stabs}, $\Stab(\eta_\lambda) = \Stab(\eta_0) \cap \Stab_W(\eta_\lambda)$.
	Finally, we have $\Mon(C) = \Stab(\eta_0) \cap \Iso(C)$.
	
Let us prove that $\Stab(\eta_0)$ is closed under the action on $\P \setminus \P_2$. Let $g \in \mathfrak{S}(M)$ preserve the orbits of $\Stab(\eta_0)$ acting on $\P\setminus \P_2$. This means that $\eta_0(\alpha) = \eta_0(g(\alpha))$ for all $\alpha \in \P\setminus \P_2$. Consider the expansion of the map $\eta_0$ in the basis $\zeta_\beta$, $\beta \in \P \setminus \P_2$,
$$
	\eta_0 = \sum_{\beta \in \P\setminus \P_2} \eta_0(\beta)\zeta_\beta.
$$
For $\alpha \in \P$,
\begin{align*}
	g^{-1}(\eta_0)(\alpha) = \eta_0(g(\alpha)) &= \sum_{\beta \in \P\setminus \P_2} \eta_0(\beta)\zeta_\beta(g(\alpha)) 
	\overset{(\ref{eq:comb:eta-alpha-g})}{=}
	\sum_{\beta \in \P\setminus \P_2} \eta_0(\beta)\zeta_{g^{-1}(\beta)}(\alpha)\\
	&=\sum_{\beta \in \P \setminus \P_2} \eta_0(g(\beta))\zeta_{\beta}(\alpha) = \sum_{\beta \in \P \setminus \P_2} \eta_0(\beta)\zeta_{\beta}(\alpha) = \eta_0(\alpha).
\end{align*}
Hence $g \in \Stab(\eta_0)$, and therefore $\Stab(\eta_0)$ is closed under the action on $\P\setminus \P_2$.
	
\vspace{1em}	
Part \ref{thm:main:c}.
Let $x \in \F(\O,\Q)$ be a function that takes equal values on each orbit and different values on different orbits of $H_2$ acting on $\O$.
From \cref{eq:comb:w-eta-x} and \cref{eq:comb:g-xi-x},
$\Stab_W(\xi_x) = \{ g \in \mathfrak{S}(\M) \mid W(g(\xi_x)) = W(\xi_x) \} = \{ g \in \mathfrak{S}(\M) \mid g(x) = x \}$.
For every $g \in H_2$, $g(x) = x$ and thus $H_2 \subseteq \Stab_W(\xi_x)$. Conversely, if $g \in \Stab_W(\xi_x)$, then $g(x) = x$ and thus $g$ preserves the orbits of $H_2$.
As in the proof of part \ref{thm:main:a}, the group $H_2$ is closed under the action on $\O$, and hence $g \in H_2$.
Finally, $H_2 = \Stab_W(\xi_x)$.

Let $\{X_1,\dots, X_t\}$ be the set of orbits of $H_1$ acting on the set $\P\setminus\P_2$. 
Let $c_1, \dots, c_t \in \Q$ be different nonnegative numbers and define the map
$$
\eta_0 = \sum_{i=1}^t c_i \sum_{\alpha \in X_i} \zeta_\alpha.
$$
For all $g \in H_1$, 
$$
g(\eta_0)= \sum_{i=1}^t c_i \sum_{\alpha \in X_i} g(\zeta_\alpha) \overset{(\ref{eq:comb:eta-alpha-g})}{=} \sum_{i=1}^t c_i \sum_{\alpha \in X_i} \zeta_{g(\alpha)}
= \sum_{i=1}^t c_i \sum_{g^{-1}(\alpha) \in X_i} \zeta_{\alpha} = 
\eta_0,
$$
and thus $H_1 \leq \Stab(\eta_0)$.
Conversely, if $g \in \Stab(\eta_0)$, then, for every $i \in \range{t}$ and every $\beta \in X_i$, $\eta_0(g^{-1}(\beta)) = g(\eta_0)(\beta) = \eta_0(\beta) = c_i$. Hence $g^{-1}(\beta) \in X_i$ and thus $g^{-1}(X_i) = X_i$. Since $H_1$ is closed under the action on $\P\setminus \P_2$, we have $g^{-1} \in H_1$ and thus $g \in H_1$.
Hence $H_1 = \Stab(\eta_0)$.

Since $\eta_0 \in \Ker W$ and $\xi_x \in V_0$, from \Cref{lemma:sum-of-maps-stabs}, $\Stab_W(\eta_0 + \xi_x) = \Stab_W(\xi_x) = H_2$ and $\Stab(\eta_0 + \xi_x) = \Stab(\eta_0) \cap \Stab_W(\xi_x) = H_1 \cap H_2$.

There exist two nonnegative integers $c, c'$ such that $c'(\eta_0 + \xi_x) + c\id_{\P_2}$ takes positive integer values.
From \Cref{lemma:plus-const}, $\Stab_W(c'(\eta_0 + \xi_x) + c \id_{\P_2}) = \Stab_W(c'(\eta_0 + \xi_x))= \Stab_W(\eta_0 + \xi_x)= H_2$ and $\Stab(c'(\eta_0 + \xi_x) + c \id_{\P_2}) = \Stab(c'(\eta_0 + \xi_x)) = \Stab(\eta_0 + \xi_x) = H_1 \cap H_2$.

From \Cref{thm:from-eta-to-parametrization}, there exists an encoding map $\lambda :\M \rightarrow A^n$ such that $\eta_{\lambda} = c'(\eta_0 + \xi_x) + c \id_{\P_2}$. Define $C = \lambda(\M)$.
	By \Cref{thm:groups-as-stabs}, $\Iso(C) = \Stab_W(\eta_\lambda) = H_2$ and $\Mon(C) = \Stab(\eta_\lambda) = H_1 \cap H_2 = H_1 \cap \Iso(C)$.
\end{proof}

\appendix
\section{Codes with $4$ elements}\label{sec:comb:m-equal-4}
There are two main reasons why the general approach of the paper fails for $m = 4$. The first reason is that for $m = 4$ the sets $\P_2$ and $\O$ have $3$ and $6$ elements correspondingly, whereas for $m \geq 5$ or $m = 3$, $\card{\P_2} = \card{\O} = \frac{m(m-1)}{2}$, which is crucial in several places in the proof. The second reason is that the matrix $B$ is not invertible and the matrix $D$ is not defined.

However, if $m = 4$ and $q \geq 3$, we still can use the basic idea and an analogue of \Cref{thm:main} holds. 
For this, let us make several changes in the text of the paper.
Replace $\P_2$ with
$$
\P_2' \defeq \Bigl\lbrace \big\lbrace \{ i \}, \{ j\}, \{\M\setminus \{i,j \} \} \big\rbrace \Bigm| \{i, j\} \subset \M \Bigr\rbrace.
$$
The set $\P_2'$ is well-defined since each of its elements has $3$ classes, which is not greater than $q$.

Replace the matrices $B$ and $D$ with $B'$ and $D'$, for $p, t \in \O$, where
$$
B'_{p,t} \defeq \left\lbrace \begin{array}{ll}
1,& \text{if } \card{p \cap t} \neq 0;\\
0,& \text{if } \card{p \cap t} = 0.
\end{array}\right.\quad
D'_{p,t} \defeq \left\lbrace \begin{array}{ll}
\frac{1}{5},& \text{if } \card{p \cap t} \neq 0;\\
-\frac{4}{5},& \text{if } \card{p \cap t} = 0.
\end{array}\right.
$$
And replace the map $\xi_x$ with
$$
\xi_x' \defeq \sum_{p = \{i,j\} \in \O} \sum_{r \in \O} D'_{r,p}x(r) \id_{ \left\lbrace \{ i \}, \{ j \} , \M\setminus p \right\rbrace }.
$$
In such a way $\card{\P_2'} = \card{\O}$ and the bijection $\P_2' \rightarrow \O, \{ \{i\}, \{j\}, \M\setminus \{i, j\} \} \mapsto \{i,j\}$ preserves the action of the group $\mathfrak{S}(\M)$. Also, $B'$ is invertible and $B'D' = D'B' = I$, where $I$ is the $6 \times 6$ identity matrix.

One can verify that all the statement, equalities and proofs remain correct with these changes.

Consider the case $m = 4$ and $q = 2$. The set $\P$ contains $8$ partitions and the set $\O$ contains $6$ pairs.
Fixing bases in $\F(\P,\Q)$ and $\F(\O, \Q)$, look at the $\Q$-linear map $W$ as the following $6 \times 8$ matrix over $\Q$:
$$
W = \left(
\begin{array}{l|cccc|ccr}
0 & 1& 1& 0& 0& 0& 1& 1\\
0 &1 &0 &1 &0 &1 &0 &1\\
0 &1 &0 &0 &1 &1 &1 &0\\
0 &0 &1 &1 &0 &1 &1 &0\\
0 &0 &1 &0 &1 &1 &0 &1\\
0 &0 &0 &1 &1 &0 &1 &1\\
\end{array}
\right).
$$
The vertical lines separate the columns labeled by partitions from three different orbits under the action of $\mathfrak{S}_4$ on the set $\P$.
\begin{proposition}
	Let $C$ be a binary code with $4$ codewords. The groups $\Iso(C) \leq \mathfrak{S}_4$ and $\Mon(C) \leq \mathfrak{S}_4$ are equal and are closed under the action on $\O$.
	For every subgroup $H \leq \mathfrak{S}_4$ that is closed under the action on $\O$, there exists a binary code $C$ with $4$ codewords such that $\Iso(C) = \Mon(C) = H$.
\end{proposition}
\begin{proof}
It is easy to calculate that $\Ker W$ is a subspace of dimension $2$ generated by $(1,0,0,0,0,0,0,0)^T$ and $(0,1,1,1,1, -1,-1,-1)^T.$ Note that $\Ker W$ is invariant under the action of $\mathfrak{S}_4$.

Let $\lambda: \M \rightarrow \{0,1\}^n$ be an encoding map of $C$. Assume that $W(\eta_\lambda) = W(\eta_{g\lambda})$, or equivalently, $\eta_\lambda - \eta_{g\lambda} \in \Ker W$, for some $g \in \mathfrak{S}_4$. In the proof of \Cref{thm:groups-as-stabs} we showed that $\eta_{g\lambda} = g(\eta_\lambda)$. Then,
$24(\eta_\lambda - \eta_{g\lambda}) = \sum_{h \in \mathfrak{S}_4} h(\eta_\lambda - g(\eta_\lambda)) = \sum_{h \in \mathfrak{S}_4} h(\eta_\lambda) - \sum_{h\in \mathfrak{S}_4}h(\eta_\lambda) = 0$. Hence, $W(\eta_\lambda) = W(\eta_{g\lambda})$ implies $\eta_\lambda = \eta_{g\lambda}$ for all $g \in \mathfrak{S}_4$.
By \Cref{thm:comb-extension-criterion} and \Cref{thm:groups-as-stabs}, $\Iso(C) = \Mon(C)$.

The fact that the groups $\Iso(C) = \Mon(C)$ are closed under the action on $\O$ is proven in the same way as \Cref{thm:main} \ref{thm:main:a}.

To prove the last statement, for each subgroup of $\mathfrak{S}_4$ closed under the action on $\O$ we give a construction of the code. Any binary code of cardinality $4$, up to a monomial equivalence, has the following form,
$$
\begin{array}{c|cccc|ccc}
x_1 & x_2& x_3& x_4& x_5& x_6& x_7& x_8\\
\hline
0 & 1& 0& 0& 0& 1& 1& 1 \\
0 & 0& 1& 0& 0& 1& 0& 0 \\
0 & 0& 0& 1& 0& 0& 1& 0 \\
0 & 0& 0& 0& 1& 0& 0& 1
\end{array},
$$
where $x_1,\dots,x_8$ are nonnegative integers; see \Cref{ex:comb:ex1} for more details.
There exist $7$, up to an automorphism, non-equivalent subgroups in $\mathfrak{S}_4$ that are closed under the action on $\O$. The groups and corresponding multiplicity functions of the corresponding codes are given in the following table.
\begin{center}\begin{tabular}{|c|c|}
	\hline 
	Group & $\eta_\lambda = (x_1,\dots,x_8)$ \\ 
	\hline 
	$\{e\}$ & (0,1,2,3,4,0,0,0) \\ 
	\hline
	$\langle (1,2) \rangle \cong \mathfrak{S}_2$ & (0,1,1,2,3,0,0,0) \\
	\hline
	$\langle (1,2), (1,2, 3)\rangle \cong \mathfrak{S}_3$ & (0,1,1,1,2,0,0,0) \\
	\hline
	$\langle (1,2), (3,4)\rangle\cong K_4$ & (0,1,1,2,2,0,0,0) \\
	\hline
	$\langle (1,2)(3,4), (1,3)(2,4) \rangle \cong K_4$ & (0,0,0,0,0,2,1,0) \\
	\hline
	$\langle (1,2,3,4), (1,3) \rangle \cong D_8$ & (0,0,0,0,0,1,0,1) \\
	\hline
	$\mathfrak{S}_4$ & (0,1,1,1,1,0,0,0) \\
	\hline
\end{tabular}
\end{center}
\end{proof}

\bibliographystyle{abbrv}
\bibliography{biblio_7}

\end{document}